\documentclass[a4paper, 11pt]{article}
\usepackage[utf8]{inputenc} 
\usepackage{hyperref} 
\usepackage{amsmath} 
\usepackage{amssymb} 
\usepackage{geometry}
\usepackage{graphicx}
\usepackage{color}
\usepackage{pdfsync}

\newenvironment{proof}{\paragraph{Proof:}}{\hfill$\square$} 
\newtheorem{theorem}{Theorem}[section] 
\numberwithin{theorem}{section} 
\newtheorem{lemma}[theorem]{Lemma}

\newtheorem{proposition}[theorem]{Proposition} 
\newtheorem{remark}[theorem]{Remark} 

\let\fd=\rightarrow
\let\lfd=\longrightarrow
\let\lfc=\longmapsto

\let\eps=\epsilon
\let\cj=\overline

\let\dps=\displaystyle

\let\bb=\bigbreak

\def\P{{\msb P}}
\def\f#1{{\msb F}_{#1}}

\def\tr{\mathop{\rm Tr}\nolimits}

\def\pt#1{\left(#1\right)}
\def\ssi{\hbox{ si et seulement si }}

\def\biq#1{\quad \hbox{ #1 }\quad }

\newcount\annee\annee=\year\advance\annee by -2000
\def\og{\leavevmode\raise.3ex\hbox{$\scriptscriptstyle\langle\!\langle$}}
\def\fg{\leavevmode\raise.3ex\hbox{$\scriptscriptstyle\rangle\!\rangle$}}

\font\tenmsb=msbm10 
\font\sevenmsb=msbm7
\font\fivemsb=msbm5
\newfam\msbfam
\def\msb{\fam\msbfam\tenmsb}%
\textfont\msbfam=\tenmsb \scriptfont\msbfam=\sevenmsb%
\scriptscriptfont\msbfam=\fivemsb%

\begin{document} 
\pagestyle{empty}
\title{Non-linearity of the Carlet-Feng function,\\ and  repartition of Gauss sums} 
\date{}
\author{Fran\c cois Rodier\thanks{ \small Aix Marseille Universit\'e, CNRS, Centrale Marseille, Institut de Math\'ematiques de Marseille, UMR 7373, 13288 Marseille, France }}

\maketitle
\begin{abstract} 
The search for Boolean functions that can withstand the main crypyographic attacks is essential.
In 2008, Carlet and Feng studied a class of functions which have optimal cryptographic properties with the exception of nonlinearity for which they give a good but not optimal bound. Carlet and some people who have also worked on this problem of nonlinearity have asked for a new answer to this problem.
We provide a new solution to improve the evaluation of the nonlinearity of the Carlet-Feng function, by means of the estimation of the distribution of Gauss sums. This work is in progress and we give some suggestions to improve this work.
\end{abstract}
\thispagestyle{empty}

\parindent=0mm
\newcommand{\F}{\mathbb F}
\newcommand{\E}{\mathcal E}
\newcommand{\K}{\mathbb K}
\newcommand{\ii}{\mathbbm{i}}
\newcommand{\mR}{{\mathcal R}}
\newcommand{\cS}{{\mathcal S}}
\newcommand{\fS}{{\mathfrak S}}
\newcommand{\Tr}{ \ensuremath{ \mathrm{Tr}}}
\renewcommand{\P}{\mathbb P}
\newcommand\numberthis{\addtocounter{equation}{1}\tag{\theequation}}


\long\def\zero #1{{}}
\def\pause{}
\let\g=\gamma
\def\cx{{\cal X}}
\def\ssi{\hbox{ if and only if }}
\def\s{{\sigma(\mu)}}


\section{Introduction}

Boolean functions on the space $\f{2}^m$ are not only important in the
theory of error-correcting codes, but also in cryptography, where they
occur in stream ciphers or private key systems.
In both cases, the properties of systems depend on
the nonlinearity of a Boolean
 function.
The   nonlinearity of a Boolean function  
$f:\f2^m\lfd\f2$ 
is
the distance from $f$ to the set of affine functions  with
$m$ variables. 
 The nonlinearity is linked to the covering radius of Reed-Muller
codes. %
 It is also an
important cryptographic parameter.
We refer to \cite{ca} for a
global survey on the  Boolean functions.

It is useful to have at one's disposal Boolean functions
 with  highest nonlinearity.  These
functions have been studied in the case where $m$ is even, and have
been called ``bent'' functions.
For these, the
degree of nonlinearity is  well known, we know how to  construct
several series of them.

The problem of the research of the maximum of the degree of
nonlinearity  comes down to  minimize the Fourier transform of
Boolean functions. 

\subsection{The Carlet-Feng function}

Let $n$ be a positive integer and $q=2^n$.
In 2008, Carlet and Feng \cite{cafe} studied a class of Boolean functions $f$ on $\f{2^n}$ which is defined by
their support
  $$\{0, 1, \alpha, \alpha^2,\dots,\alpha^{2^{n-1}-2}\}$$
  where $\alpha$ is a primitive element of the field $\f{2^n}$.
In the same article they show that
these functions when $n$ varies  have
 optimum algebraic immunity,
good nonlinearity and optimum
algebraic degree.
These computations are very good but still not good enough: in fact these bounds are
 not enough for ensuring a sufficient nonlinearity. 
Some works have been done on that by Q. Wang and P. Stanica \cite{wast} and other authors (cf. Li et al \cite{lcz} and Tang et al. \cite{tct}).
They find the bound
$$2^{n-1}-  nl (f)  \le {1\over\pi}q^{1/2} \pt{ n \ln 2 + \gamma  +\ln\pt{8\over\pi} +o (1) }$$
where $\gamma$ is the Euler's constant.
  Nevertheless, there is a gap between the bound that they can prove and the actual
computed values  for a finite numbers of functions which
are very good, of order
$2^{n-1} - 2^{n/2}$. 
Carlet and 
some  authors cited above \cite{lcz,tct,wast} who have also worked on this nonlinearity asked for new answer to this problem. 
In this paper
we bring a new solution to improve the evaluation of the nonlinearity of the Carlet-Feng function, by means of the estimation of the distribution of Gauss sums. 
We will find a slightly better asymptotic bound (see (\ref{resultat})) but  this work is in progress and we give some suggestions to improve this work and hopefully to get a result closer to what expected.
It will be the same for other classes of Boolean functions which are based on Carlet-Feng construction.

\subsection{The nonlinearity}

The nonlinearity of these functions is given by   
 \begin{equation}
\label{nl}
nl(f)=2^{n-1}- \max_{\lambda\in\f{2^n}^*} |S_\lambda |
\biq{where}
S_\lambda = \sum_{i=2^{n-1}-1}^{2^n-2} (-1)^{\tr(\lambda \alpha^i)}.
\end{equation}
We define $\zeta= \exp\pt{2i\pi\over 2^n-1}$, $\chi$ be the multiplicative
character of $\f{2^n}$ such that
$\chi(\alpha)=\zeta$.
For $a\in \f q^*$
let us define the Gaussian sum $G(a,\chi)$  by
$$G(a,\chi)=\sum_{x\in\f q^*} \chi(x) \exp( \pi i \tr(ax))$$
and $G (\chi)= G (1,\chi)$.
By Fourier transformation of (\ref{nl}) we get
  $$S_\lambda={1\over q-1}\pt{ \sum_{\mu=1}^{q-2} G(\chi^{\mu})\zeta^{-\mu\ell}{\zeta^{-\mu({q\over2}-1)}-1\over 1-\zeta^{-\mu}} -{q\over2}}.$$

Carlet and Feng deduced from that the bound
 $$|S_\lambda | \le  {1\over q-1}\pt{ \sum_{\mu=1}^{q-2} \sqrt q  \left|{\zeta^{-\mu({q\over2}-1)}-1\over 1-\zeta^{-\mu}} \right| +{q\over2}}.$$
 The upperbound of $|S_\lambda|$ is attained if
 the arguments of 
$G(\chi^{\mu})\zeta^{-\mu\ell}$ are the opposite of the ones 
of ${\zeta^{-\mu({q\over2}-1)}-1\over 1-\zeta^{-\mu}}$.
I will show that this situation is impossible and that will lead us to a better bound.

\section{Equidistribution of the arguments of Gauss sums}
\subsection{A result of Nicolas Katz}

Nicolas Katz (chapter 9 in \cite{katz})
has proved that
\begin{proposition}
\label{equid}
For $a$ fixed in $\f{2^n}^*$
the arguments of $ G(a,\chi^\mu) $ for $1\le \mu\le q-2$ are equidistributed on the segment $[-\pi,\ \pi]$.
\end{proposition}

For $l$ fixed in $\f{2^n}^*$
the arguments of $G(\chi^\mu)\zeta^{-\mu l}$ for $1\le \mu\le q-2$ are also equidistributed on the segment $[-\pi,\ \pi]$ since
by \cite{nied} theorem 5.12, they satisfy: $G(\chi^\mu)\zeta^{-\mu l}= G( (-1)^{{\tr(\alpha^l )}}, \chi^\mu )$.
\subsection{Discrepancy}

To get a result a little more precise than Katz's we need the notion of discrepancy.
We define the
discrepancy (see \cite{drti} or \cite{kuni})
of a
sequence of $N$ real numbers 
$x_1, \dots,x_N\in [0,\ 1[$ 
 by
$$D_N (x_N)= \max_{0\le x \le 1} |{A(x, N)\over N} -x |$$
where
$A(x, N) $ = number of $m\le N$ such that $x_m\le x$.

\begin{proposition} \label{ud} A sequence $(x_N)_{N\ge1}$ is uniformly distributed mod 1 if and only if 
$$\lim_{N\fd\infty} D_N(x_N) = 0.$$ 
\end{proposition}

We have an estimate of the discrepancy thanks to Erd\"os-Turan-Koksma's inequality.

\begin{lemma} [Erd\"os-Turan-Koksma's inequality]
There is an absolute constant $C$ such that for every $H\ge1$, 
$$D_N(x_N)<C\pt{{1\over H} +\sum_{h=1}^H {1\over h}\left| {1\over N}\sum_{m=1}^N \exp(2\pi i h x_m) \right| }$$
\end{lemma}

We will use also a result of Deligne obtained by using Algebraic Geometry ``\`a la Grothendieck''.
\begin{proposition}[Deligne \cite{deli}]
For $\psi$ an additive character  of $\f q$ and $a\in\f q^*$, we have
$$|\sum_{x_1x_2\dots x_r=1} \psi(x_1+x_2+\cdots+x_r) |\le rq^{(r-1)/2}.$$
 \end{proposition}

With this proposition, we can show that, for $a\ne0$ one has
$ |\sum_{1\le \mu\le q-2}  G(a,\chi^\mu)^r|  
 \le 1+ r q^{(r+1)/2}.$
So we can show more than Katz's result with the help of proposition (\ref{ud}).

 \begin{proposition}
\label{disc}
For $l$ fixed in $\f{2^n}^*$
the arguments $\arg (z_\mu)$ of $z_\mu=G(\chi^\mu)\zeta^{-\mu l}$ for $1\le \mu\le q-2$ 
fulfill
$$D_{q-2}\pt{ \arg (z_\mu) \over 2\pi} < O (q^{-1/4})$$
\end{proposition}
\begin{proof}
We use Erd\"os-Turan-Koksma's inequality to evaluate this dicrepancy, and use Deligne's result to bound 
$ |\sum_{1\le \mu\le q-2}  G(a,\chi^\mu)^r|  $ 
which gives the result.
Whence, if $H\le q^{1/2}$ 
\begin{eqnarray*}
D_{q-2} \pt{ \arg (z_\mu) \over 2\pi} &<&O\pt{{1\over H} +{1\over q-2}\sum_{h=1}^H {1\over h q^{h/2}}\left| \sum_{\mu=1}^{q-2} G( (-1)^{{\tr(\alpha^l )}}, \chi^\mu )^h\right| }\\
&<&O\pt{{1\over H} + {1\over q-2} \sum_{h=1}^H {1\over h q^{h/2}} hq^{(h+1)/2} }\cr
&<&O\pt{{1\over H} +{H q^{1/2} \over q-2}}
\end{eqnarray*}

If $H=q^{1/4}$, then
$
D_{q-2}  \pt{ \arg (z_\mu) \over 2\pi} 
<O\pt{{{q^{3/4}} +q^{3/4}\over q-2}} =O\pt{q^{-1/4}}$.
\end{proof}

\begin{lemma}
\label{interval}
If the $a_m$ is an increasing sequence and if the discrepancy of
 $a_m$ is $D$, then $|a_i-{i\over m}| \le D $.
\end{lemma}

Let
$A(I, N) $ = number of $m\le N$ such that $x_m\in I$.
Let $I_1$ the interval $[0,\ {i\over m}-D[$ and  $I_\eps$ the interval $[0,\ {i\over m}-D-\eps]$ where $\eps$ is a positive real number, then
$$\bigg| {A(I_\eps,m) \over m}-({i\over m}-D-\eps)\bigg| \le D$$
hence
$${A(I_\eps,m) } \le mD+m({i\over m}-D-\eps)=i-m\eps$$
therefore the interval $I_\eps$ contains less than $i-m\eps$ elements, and does not contains $a_i$ which is the $i$-th elements in the sequence. Therefore, since we can take $\eps$ as small as we want one has
$a_i \notin I_\eps$.

In the same way let
 $I_2$ the interval  $[0,\ {i\over m}+D]$, then
$$ \bigg| {A(I_2,m) \over m}-({i\over m}+D)\bigg| \le D$$
hence
$$i\le m( {i\over m}+D)-D \le A(I_2,m) $$
therefore
$a_i\in I_2 -I_1=\ [{i\over m}-D ,\ {i\over m}+D ]$.

\section{Distribution of the arguments of $a_\mu$}

Let
$$a_\mu={\zeta^{-\mu({q\over2}-1)}-1\over 1-\zeta^{-\mu}}$$

\begin{proposition}
The $a_\mu$ are on the singular plane cubic
which is the image of  the unit circle by the map
$$z \fd {1\over z+z^2}$$
with $|z| = 1$.
 The absolute value is
$|a_\mu |= (2 \cos ({\pi\mu\over 2(q-1)}))^{-1} $.
The argument is
 $\arg a_\mu=   {3\pi\mu \over2 ( q-1)}$ for $\mu$ even or  $\pi/2 + {3\pi\mu \over2 ( q-1)}$ for $\mu$ odd.
The complex conjugate of $a_\mu$ is $a_{q-1-\mu}$.

\end{proposition}

\begin{proof}

If $\mu$ is even,  let us take  $z=\exp(-{\pi\mu i\over q-1})$.
One has $z^2=\zeta^{-\mu}$.
And one has also
$$z^{q-1}= \exp(-{\pi\mu i})= \exp(-2{\pi i\mu/2})=1.
$$
Thus
$z^{q-2} =z^{-1}$,
hence
$$a_\mu={z^{2({q\over2}-1)}-1\over 1-z^2}={z^{({q}-2)}-1\over 1-z^2}={z^{-1}-1\over 1-z^2}
={1 -z\over z-z^3}={1\over z+z^2}
$$
If $\mu$ is odd,  let us take 
\begin{eqnarray*}
z=-\exp(-{\pi\mu i\over q-1})=\exp(i\pi-{\pi\mu i\over q-1})=\exp(-{i\pi\over q-1}(\mu + q-1))
\end{eqnarray*}
Then
$q\le (\mu + q-1)\le 2q-2$.
And we still have
$$a_\mu={z^{2({q\over2}-1)}-1\over 1-z^2}={z^{({q}-2)}-1\over 1-z^2}={z^{-1}-1\over 1-z^2}
={1 -z\over z-z^3}={1\over z+z^2}
$$

So the set of $ a_\mu $ is on a cubic with double point of complex parametric equation
$$z\lfc {1\over z+z^2}$$
for $ | z | = 1 $.

Now we consider the lozenge of vertices $0,z,z+z^2,z^2$. For $\mu$ even the angle between the axis $Ox$ and $z$ is  the same as between $ z $ and $ z ^ 2 $ and is $-{\pi\mu \over q-1}$.
The absolute value of $ z + z ^ 2 $ is the length of the diagonal $ 0, z + z ^ 2 $. It is easy to find  $2\cos{\pi\mu \over 2(q-1)}$. 
The angle between the $x$-axis and $ z + z ^ 2 $ is $-{3\pi\mu \over2( q-1)}$.
The angle between the $x$-axis and ${1\over z+z^2}$ is ${3\pi\mu \over2( q-1)}$.
For $\mu$ odd, the reasoning is the same.

\end{proof}



\section{Applications}

 So we conclude from the preceding sections that for a fixed $\ell$
 the arguments of $G(\chi^{\mu})\zeta^{-\mu\ell}$ are equidistributed on $[-\pi,\ \pi]$, and
 the arguments of $ a_\mu $ are equidistributed on $[-3\pi/2,\  3\pi/2]$ so, as we said before,
it is impossible to have
$\arg ( G(\chi^{\mu})\zeta^{-\mu\ell}) + \arg (a_\mu)=0 \pmod{2\pi}$
 and the upperbound of $|S_\lambda|$ is not attained.

The preceding proposition implies
$$ \sum_{\mu=1}^{q-2} G(\chi^{\mu})\zeta^{-\mu\ell}a_\mu 
\le  2\max_{\sigma} \bigg(\Re e\sum_{ \scriptstyle{ \mu=2\atop \mu \hbox{\tiny even}}}^{q-2} (\cj{h_{\sigma(\mu)}} a_\mu)\bigg)
$$
where $\{h_\mu\}$ is the set of   Gauss sums 
 and $\sigma$ is some permutation of this set.
Let us number increasingly the $h_\mu$ (with multiplicities) for $\mu$ even from $0$ to $2\pi$. 
Let $k_x =q^{1/2}\exp\pt{ i\pt{{2\pi x\over q-1}}}$.

\begin{lemma}
\label{ordre}
For $2\le \mu\le q-2$ and $\mu$ even, we have
 \begin{eqnarray*}
\bigg| \Re e(\cj{h_\s} a_\mu-\cj{k_\s} a_\mu)\bigg|
=O\pt{q^{1/4}\over \cos {\pi\mu\over2(q-1)}}
\end{eqnarray*}
\end{lemma}
\def\s{{\sigma(\mu)}}
\begin{proof}
 We have
\begin{eqnarray*}
\lefteqn{ \Re e(\cj{h_\s} a_\mu-\cj{k_\s} a_\mu)}\\
&=& -q^{1/2}{(\cos{ {( 3\pi\mu\over 2(q-1)}-\arg h_\s})-(\cos{ {( 3\pi\mu\over 2(q-1)}- \arg k_\s})\over 2\cos { \pi\mu \over 2(q-1)}}\\
&=& -q^{1/2}{(\sin{ \arg k_\s-\arg h_\s\over2})\sin( {( 3\pi\mu\over 2(q-1)} - { \arg k_\s+\arg h_\s\over2} )\over  \cos { \pi\mu \over 2(q-1)}}\\
\end{eqnarray*}
Thus
\begin{eqnarray*}
\bigg|\Re e(\cj{h_\s} a_\mu-\cj{k_\s} a_\mu)\bigg|
&\le& q^{1/2} 
{|(\sin{ \arg k_\s-\arg h_\s\over2})|\over  \cos { \pi\mu \over 2(q-1)}}\\
&\le& q^{1/2}
{|({ \arg k_\s-\arg h_\s})|\over 2\cos { \pi\mu \over 2(q-1)}}\\
\noalign{\hbox{from the proposition \ref{disc} and the lemma \ref{interval}:}}
&\le& q^{1/2} 
{O (q^{-1/4})\over 2\cos { \pi\mu \over 2(q-1)}}\\
&\le& O (q^{1/4} )
{1\over 2\cos { \pi\mu \over 2(q-1)}}.
\end{eqnarray*}
\end{proof}

From the proposition \ref{disc}, we get the following lemma.

\begin{lemma}
  The sums $ \dps \Re e\sum_{ \scriptstyle{ \mu=2\atop \mu \hbox{\tiny even}}}^{q-2} (\cj{h_{\sigma(\mu)}} a_\mu)$
satisfy
\begin{eqnarray*}
\lefteqn{\max_{\sigma} \bigg(\Re e\sum_{{ \mu=1}}^{q-2} (\cj{h_{\sigma(\mu)}} a_\mu)\bigg) }\\
 &\le&
 2 \Re e\sum_{ \scriptstyle{ \mu=2\atop \mu \hbox{\tiny even}}}^{q/2} (\cj{k_{\mu/2}} a_\mu)
+2\Re e \sum_{ \scriptstyle{ \mu=q/2\atop \mu \hbox{\tiny even}}}^{2q/3}  (\cj{k_{3\mu/2-q/2}} a_\mu)
+2\Re e  \sum_{ \scriptstyle{ \mu=2q/3\atop \mu \hbox{\tiny even}}}^{q-2} (\cj{k_{3\mu/4}} a_\mu) +O(q^{5/4}\log q).
\end{eqnarray*}
\end{lemma}
\begin{proof}

We first have from the lemma \ref{ordre}
$$\max_{\sigma} \bigg(\Re e\sum_{{ \mu=1}}^{q-2} (\cj{h_{\sigma(\mu)}} a_\mu)\bigg)
=\max_{\sigma} \bigg(\Re e\sum_{{ \mu=1}}^{q-2} (\cj{k_{\sigma(\mu)}} a_\mu)\bigg)+O(q^{5/4}\log q)$$
because
$$\bigg| \Re e\sum_{{ \mu=1}}^{q-2} (\cj{h_{\sigma(\mu)}} a_\mu)
-\Re e\sum_{{ \mu=1}}^{q-2} (\cj{k_{\sigma(\mu)}} a_\mu)\bigg| =O\pt{q^{1/4}\Re e\sum_{{ \mu=1}}^{q-2} {1\over \cos {\pi\mu\over2(q-1)}}}$$
and \cite{cafe} gives an upper bound of the last sum.

We denote by $b_\mu$ the following numbers for $\mu$ even and $2\le \mu\le q-2$:
if $2\le\mu\le q/2$, then $b_\mu=k_{\mu/2}$, if $q/2<\mu\le 2q/3$, then $b_\mu=k_{3\mu/2-q/2}$, if $2q/3<\mu\le q-2$, then $b_\mu=k_{3\mu/4}$. So that we have
$$
{2 \Re e\sum_{ \scriptstyle{ \mu=2\atop \mu \hbox{\tiny even}}}^{q/2} (\cj{k_{\mu/2}} a_\mu)
+2\Re e \sum_{ \scriptstyle{ \mu=q/2\atop \mu \hbox{\tiny even}}}^{2q/3}  (\cj{k_{3\mu/2-q/2}} a_\mu)
+2\Re e  \sum_{ \scriptstyle{ \mu=2q/3\atop \mu \hbox{\tiny even}}}^{q-2} (\cj{k_{3\mu/4}} a_\mu) }
=2 \Re e\sum_{ \scriptstyle{ \mu=2\atop \mu \hbox{\tiny even}}}^{q/2} (\cj{b_\mu} a_\mu).
$$

Then we  want to compare the sum
 $\dps\Re e\sum_{ \scriptstyle{ \mu=2\atop \mu \hbox{\tiny even}}}^{q-2} (\cj{k_{\sigma(\mu)}} a_\mu)$ with the sum
  $\dps\Re e\sum_{ \scriptstyle{ \mu=2\atop \mu \hbox{\tiny even}}}^{q-2} (\cj{b_{\mu}} a_\mu)$.
Let $ \beta $ be the  largest integer such that
$k_{\sigma(\beta)} \ne b_\beta$.
Let $ \tau $ be the transposition between $ \beta $ and $ \sigma (\beta) $.
Then  one can check that
$$\Re e (b_\beta a_\beta + k_{\sigma(\beta)} a_{\sigma(\beta)} )
> \Re e ( k_{\sigma(\beta)}  a_\beta + b_\beta a_{\sigma(\beta)} )$$
therefore
$$ 2\Re e\sum_{ \scriptstyle{ \mu=1\atop \mu \hbox{\tiny even}}}^{q-2} (\cj{k_{\sigma(\mu)}} a_\mu)
< 2\Re e\sum_{ \scriptstyle{ \mu=1\atop \mu \hbox{\tiny even}}}^{q-2} (\cj{k_{\sigma\circ\tau(\mu)}} a_\mu).
$$
Thus, if there exists such a $\beta$, the sum is not maximal.
\end{proof}

\begin{remark}
The condition $\Re e (b_\beta a_\beta + k_{\sigma(\beta)} a_{\sigma(\beta)} )
> \Re e ( k_{\sigma(\beta)}  a_\beta + b_\beta a_{\sigma(\beta)} )$
is equivalent to
$\Re((b_\beta  - k_{\sigma(\beta)})(a_\beta - a_{\sigma(\beta)}))>0$, which may be easier to check.
\end{remark}

\bb
Then we  consider the sum
 $\Re e\sum_{ \scriptstyle{ \mu=1\atop \mu \hbox{\tiny even}}}^{q-2} (\cj{b_\mu} a_\mu)$.
 For $2\le\mu\le q/2$ and $\mu$ even, then $b_\mu=k_{\mu/2}$, which imply that these $b_\mu$'s 
 form a set with $q/4$ elements  uniformly distributed in the interval $[0,\ q/4]$. 
  For $q/2<\mu\le 2q/3$ and $\mu$ even, then $b_\mu=k_{3\mu/2-q/2}$, which imply that these $b_\mu$'s 
 form a set with $[q/12]$ elements  uniformly distributed in the interval $[q/4,\ q/2]$. 
  For $2q/3<\mu\le q-2$ and $\mu$ even, then $b_\mu=k_{3\mu/4}$, which imply that these $b_\mu$'s 
 form a set with $[q/6]$ elements  uniformly distributed in the interval $[q/2,\ 3q/4]$. 
 Let $B$ be the set of all $b_\mu$'s.
 
Now we have to take also in consideration the $\mu$ odd. When you make the same reasoning, you end up with a set $\cj B$ which is just the complex conjugate of $B$. When you take the union $B\cup \cj B$, you get $q$ elements uniformly distributed in the interval $[0,\ 2\pi]$. 
 \bb

\begin{proposition}
 The upper bound of
$\dps\sum_{\mu=1}^{q-2} G(\chi^{\mu})\zeta^{-\mu\ell}a_\mu $
is at most equal to 
$$\dps{q^{3/2}\over\pi}(\ln q-0.3786+o(1)).$$
\end{proposition}
\begin{proof}
Up to $O(q^{5/4}\log q)$ it is enough to compute:
\begin{eqnarray*}
\max_{\sigma} \bigg(\Re e\sum_{{ \mu=1}}^{q-2} (\cj{k_{\sigma(\mu)}} a_\mu)\bigg) 
&\le&2 \Re e\sum_{ \scriptstyle{ \mu=1\atop \mu \hbox{\tiny even}}}^{q/2} (\cj{k_{\mu/2}} a_\mu)
+2\Re e \sum_{ \scriptstyle{ \mu=q/2\atop \mu \hbox{\tiny even}}}^{2q/3}  (\cj{k_{3\mu/2-q/2}} a_\mu)
+2\Re e  \sum_{ \scriptstyle{ \mu=2q/3\atop \mu \hbox{\tiny even}}}^{q-2} (\cj{k_{3\mu/4}} a_\mu) \\
&\le& 2q^{1/2} \sum_{ \scriptstyle{ \mu=1\atop \mu \hbox{\tiny even}}}^{q/2} {1\over2}
-2 {q^{1/2}   \sum_{ \scriptstyle{ \mu=q/2\atop \mu \hbox{\tiny even}}}^{2q/3}  {\cos{3\pi\mu\over2(q-1)} \over 2\cos{\pi\mu\over2(q-1)}}}
+2 {q^{1/2}  \sum_{ \scriptstyle{ \mu=2q/3\atop \mu \hbox{\tiny even}}}^{q-2} {1\over 2\cos{\pi\mu\over2(q-1)}}}\\
&\le& {q^{3/2} \over4}
-2 {q^{1/2}   \sum_{ \scriptstyle{ \mu=q/2\atop \mu \hbox{\tiny even}}}^{2q/3}\pt{ 2\cos^2 {\pi\mu\over2(q-1)}-3/2}}
+2 {q^{1/2}  \sum_{ \scriptstyle{ \mu=2q/3\atop \mu \hbox{\tiny even}}}^{q-2} {1\over 2\cos{\pi\mu\over2(q-1)}}}\\
&\le&{q^{3/2} \over2}
- 4 {q^{1/2}   \sum_{ \scriptstyle{ \mu=q/2\atop \mu \hbox{\tiny even}}}^{2q/3} \cos^2 {\pi\mu\over2(q-1)}}
+2 {q^{1/2}  \sum_{ \scriptstyle{ \mu=2q/3\atop \mu \hbox{\tiny even}}}^{q-2} {1\over 2\cos{\pi\mu\over2(q-1)}}}.
\end{eqnarray*}

 Since the function ${1\over2\cos(x\pi/2)}-{1\over\pi(1-x)}$ is continuous 
on $[2/3,\ 1]$, and since the ${\mu\over q-1}$ are uniformly distributed on $[2/3,\ 1]$ we get by \cite[theorem 1.1]{kuni}:
\begin{eqnarray*}
&&{2\over q-2} \sum_{ \scriptstyle{ \mu=2q/3\atop \mu \hbox{\tiny even}}}^{q-2}  {1\over 2\cos { \pi\mu \over 2(q-1)}}-
{2\over\pi} \sum_{ \scriptstyle{ \mu=2q/3\atop \mu \hbox{\tiny even}}}^{q-2}  {1\over q-\mu}\\
&=&{(1+o (1))}\int_{ 2/3}^{1} \pt{ {1\over2 \cos {x \pi\over 2}} -{1\over\pi} {1\over 1-x}}dx\\
&=&(1+o(1))\left[ ((\ln(\sin(1/2\pi x) + 1) - \ln(-\sin(1/2\pi x) +1))/2\pi + \log(1-x)/\pi \right]_{2/3}^1\\
&=&{\ln2\ - \ln\pi +\ln3\over \pi}- {\ln(7+4\sqrt3) \over2\pi} +o(1).
\end{eqnarray*}

Then, using Euler's formula on harmonic series:
\begin{eqnarray*}
\lefteqn{{2\over q^{1/2}(q-2)}\Re e\sum_{ \scriptstyle{ \mu=2q/3\atop \mu \hbox{\tiny even}}}^{q-2} (\cj{\sigma(h_\mu)} a_\mu)}\\
&\le& 
\pt{{2\over q-2} \sum_{ \scriptstyle{ \mu=2q/3\atop \mu \hbox{\tiny even}}}^{q-2}  {1\over 2\cos { \pi\mu \over 2(q-1)}}-
{2\over\pi} \sum_{ \scriptstyle{ \mu=2q/3\atop \mu \hbox{\tiny even}}}^{q-2}  {1\over q-\mu} }
+{2\over\pi} \sum_{ \scriptstyle{ \mu=2q/3\atop \mu \hbox{\tiny even}}}^{q-2}  {1\over q-\mu}\\
&\le& { \log q-\ln\pi +\gamma\over \pi} -{\ln(7+4\sqrt3) \over2\pi} +o(1).
\end{eqnarray*}

Now it is easy to compute 
$$\sum_{ \scriptstyle{ \mu=q/2\atop \mu \hbox{\tiny even}}}^{2q/3} \cos^2 {\pi\mu\over2(q-1)}
= {q\over12} { (\pi + (3\sqrt3) - 6)\over12 \pi}(1+o(1)).$$
Finally, 
 the upper bound of
$\sum_{\mu=1}^{q-2} G(\chi^{\mu})\zeta^{-\mu\ell}a_\mu $
is at most equal to 
\begin{eqnarray*}
\lefteqn{{q^{3/2} \over2}
- 4 {q^{1/2}   \sum_{ \scriptstyle{ \mu=q/2\atop \mu \hbox{\tiny even}}}^{2q/3} \cos^2 {\pi\mu\over2(q-1)}}
+2 {q^{1/2}  \sum_{ \scriptstyle{ \mu=2q/3\atop \mu \hbox{\tiny even}}}^{q-2} {1\over 2\cos{\pi\mu\over2(q-1)}}}+O(q^{5/4}\log q)}\\
&=&{q^{3/2} \over2}
- 4 {q^{1/2}  {q\over12} { (\pi + (3\sqrt3) - 6)\over12 \pi}
+q^{3/2}{ \log q-\ln\pi- {\ln(7+4\sqrt3) \over2} +\gamma \over \pi}+o(q^{3/2})}\\
&=&
{q^{3/2}\over\pi}\pt{ \log q-\ln\pi- {\ln(7+4\sqrt3) \over2} +\gamma +\pi/2-\pi/36-\sqrt3/12+1/6+o(1)}\\
&<&  	{q^{3/2}\over\pi}(\ln q-0.3786+o(1)).
\end{eqnarray*}

\end{proof}

\subsection{Final result}

We get finally
\begin{theorem}
The nonlinearity of the Carlet-Feng function fulfills
\begin{equation}
\label{resultat}
2^{n-1}-nl(f)\le {q^{1/2} \over\pi}\pt{ \log q -0.3786+o(1) }.
\end{equation}
\end{theorem}
  
  
\section{Conclusion}

The improvement is not very important, but this argument may be optimised by
\begin{itemize}
\item
taking in account the invariance of Gauss sums under the Frobenius automorphism;
\item
making it possible to make our  argument work for all $n$ instead of having an asymptotic result;
\item
taking in account the irregularity of the distribution of Gauss sums
(one way to do this might be to look at the equidistribution of several Gauss sums simultaneously);
\item
improving the bound of nonlinearity for other classes of Boolean functions which are based on Carlet-Feng construction.
\end{itemize}


\begin{thebibliography}{99}
 
\bibitem{ca} Claude Carlet, 
 {\sl Boolean Functions for Cryptography and Error  Correcting Codes}, Chapter of the monography,  Boolean Models and Methods in Mathematics, Computer Science and Engineering published by Cambridge University Press, Yves Crama and Peter L. Hammer (eds.), pp. 257-397, 2010.
 
\bibitem{cafe} Claude Carlet, Keqin Feng,
An infinite class of balanced functions with optimal algebraic immunity, good immunity to fast algebraic attacks and good nonlinearity. Advances in cryptology- ASIACRYPT 2008, 425-440, Lecture Notes in Comput. Sci., 5350, Springer, Berlin, 2008.

\bibitem{deli}
Deligne, P., Applications de Ia formule des traces aux sommes trigonometriques, in: Cohomologie Etale (SGA 4 1/2), Lecture Notes in Mathematics, vol. 569, Springer-Verlag.

 \bibitem{drti}
 M. Drmota and R. Tichy, Sequences, discrepancies and applications, Springer-Verlag,
Berlin, 1997.

\bibitem{katz} Katz, N.: Gauss Sums, Kloosterman Sums and Monodromy Groups, Annals of math.
Studies 116, Princeton Univ. Press, 1988

 \bibitem{kuni} L. Kuipers and H. Niederreiter, Uniform distribution of sequences, Wiley-Interscience, New York-London-Sydney, 1974.

 \bibitem{lcz} 
 Jiao Li , Claude Carlet , Xiangyong Zeng , Chunlei Li, Lei Hu , Jinyong Shan,
{\sl Two constructions of balanced Boolean functions with optimal algebraic immunity, high nonlinearity and good behavior against fast algebraic attacks} Des. Codes Cryptogr. 76 (2015), no. 2, 279-305. 

\bibitem{nied}
R. Lidl, and H. Niederreiter, {\em Introduction to finite fields and their applications.} Cambridge university press, 1994.

 \bibitem{tct} Tang D., Carlet C., Tang X.
 {\sl Highly nonlinear Boolean functions with optimal algebraic immunity and good behavior against fast algebraic attacks.}  IEEE Trans. Inform. Theory 59 (2013), no. 1, 653-664.

\bibitem{wast} Qichun Wang, Pantelimon Stanica, Trigonometric Sum Sharp Estimate and New Bounds on the Nonlinearity of Some Cryptographic Boolean Functions, Des. Codes Cryptogr. 87 (2019), no. 8, 1749-1763.

\end{thebibliography}
\end{document}